\documentclass[twocolumn,natbib]{svjour3}
\smartqed  

\usepackage{amssymb}

\usepackage{amsmath}
\spnewtheorem{prop}{Proposition}{\bf}{\it}
\usepackage{float}
\usepackage{color}

\usepackage{multirow}
\usepackage{amsfonts}
\usepackage{amssymb}
\usepackage{algorithm,algorithmic}
\usepackage{pgfplots}

\newcommand{\todo}[1]{}
\renewcommand{\todo}[1]{{\bf{ \color{red} TODO: {#1}}}}

\newcommand{\nocolor}[1]{{\color{black} #1}}

\begin{document}




\title{Calibrations Scheduling Problem with Arbitrary Lengths and Activation Length
\thanks{This work has been supported by the ALGONOW project of the THALES program and the Special Account for Research Grants of National and Kapodistrian
University of Athens,
by NSFC (No. 61433012), Shenzhen research grant (KQJSCX20180330170311901, JCYJ20180305180840138 and GGFW2017073114031767).
}}



\author{Eric Angel			\and
        Evripidis Bampis		\and
        Vincent Chau			\and
        Vassilis Zissimopoulos 
}


\institute{
Eric Angel			\at
IBISC, University Paris Saclay, Evry, France\\
\email{eric.angel@univ-evry.fr} 
\and
Evripidis Bampis		\at
Sorbonne Universit\'e, CNRS, Laboratoire d'Informatique de Paris 6, LIP6, F-75005 Paris, France
\email{Evripidis.Bampis@lip6.fr} 
\and
Vincent Chau \at          
Shenzhen Institutes of Advanced Technology, Chinese Academy of Sciences, Shenzhen, China\\
\email{vincentchau@siat.ac.cn}
\and
Vassilis Zissimopoulos \at
Department of Informatics \& Telecommunications,
National and Kapodistrian University of Athens, Athens, Greece\\
\email{vassilis@di.uoa.gr} 
}

\date{Received: date / Accepted: date}
\maketitle

\begin{abstract}
Bender et al. (SPAA 2013) proposed a theoretical framework for testing in contexts where safety mistakes must be avoided.
Testing in such a context is made by machines that need to be calibrated in a regular basis. Since calibrations have a non-negligible cost, it is important to study policies minimizing the total calibration cost while performing all the necessary tests. We focus on the single-machine setting, and we study the complexity status of different variants of the problem. First, we extend the model by considering that the jobs have arbitrary processing times and we propose an optimal polynomial time algorithm when preemption of jobs is allowed. Then, we study the case where there are many types of calibrations with their corresponding lengths and costs. We prove that the problem becomes NP-hard for arbitrary processing times even when the preemption of the jobs is allowed. Finally, we focus on the case of unit processing time jobs, and we show that a more general problem, where the recalibration of the machine is not instantaneous, can be solved in polynomial time via dynamic programming. 
\end{abstract}
%
%

%

\section{Introduction}

The scheduling problem, whose objective is to minimize the number of calibrations, was introduced by \cite{bender2013efficient}. It is motivated by the Integrated Stockpile Evaluation (ISE) program~(\cite{ise}) at Sandia National Laboratories for testing nuclear weapons in contexts where safety mistakes may have serious consequences. This motivation can be extended to the machines that need to be calibrated carefully to ensure the accuracy. Calibrations have extensive applications in several areas, including robotics (\cite{bernhardt1993robot, evans1982method, nguyen2013new}), pharmaceuticals (\cite{forina1998multivariate, bansal2004qualification}), and digital cameras  (\cite{baer2005self, barton2006sensor, zhang2002method}).
Formally, the problem can be stated as follows:  given a set $\cal{J}$ of $n$ jobs (tests), where each job $j$ is characterized by its release time $r_j$, its deadline $d_j$ and its processing time $p_j$. Each job must be processed inside $[r_j,d_j)$. We are also given a (set of) testing machine(s) that must be calibrated regularly. Calibrating a machine incurs a unit cost, and it is instantaneous, i.e., a machine can be calibrated between the execution of two jobs that are processed consecutively. A machine stays calibrated for $T$ time-units, and a job can only be processed during an interval where the machine is calibrated. The goal is to find a feasible schedule performing all the tests (jobs) between their release times and deadlines and minimizing the number of calibrations. \cite{bender2013efficient} studied the case of \emph{unit-time} jobs. They considered both the single-machine and multiple-machine problems. For the single-machine case, they showed it could be solved in polynomial time: their algorithm is called the Lazy Binning. For the multiple-machine case, they proposed a 2-approximation algorithm. However, the complexity status of the multiple-machine case with unit-time jobs remained open. Very recently, \cite{ChenLL019} gave a polynomial-time algorithm when the number of machines is constant. However, the running time grows exponentially with the number of machines. They gave a PTAS for an arbitrary number of machines.

\cite{FinemanS15} studied a first generalization of the problem by considering that the jobs have arbitrary processing times. They considered the multiple-machine case where jobs cannot be interrupted once it has been started. Since the feasibility problem is NP-hard, they considered a {\em resource-augmentation}  version of the problem (\cite{KalyanasundaramP00}). They were able to relate this version with the classical {\em machine-minimization} problem (\cite{PhillipsSTW02}) in the following way. Suppose there is an $s$-speed $\alpha$-approximation algorithm for the machine-minimization problem, then there is a $O(\alpha)$-machine $s$-speed $O(\alpha)$-approximation for the resource-augmentation version of the problem of minimizing the number of calibrations.  

Other objectives have been studied in the literature. One of the variants is studied by \cite{ChauLMW17}; they considered the flow time problem with calibration constraints (the flow time of a job is the length of the duration from its release time until its completion). They investigated the online version in which the goal is to minimize the total (weighted) flow time as well as the total cost of the calibrations. They proposed several constant competitive online algorithms where jobs are not known in advance and a polynomial-time algorithm for the offline case. \cite{Wang18} studied the time slot cost variant. Scheduling a job incurs a different cost for each time slot. \cite{Wang18} considered the case of jobs with uniform processing time, and the goal is to minimize the total cost incurred by the jobs with a limited number of calibrations. \cite{ChauFLWZ019} studied the throughput variant of the calibration scheduling problem. This variant is, in fact, a generalization of the calibration minimization problem. If this variant can be solved optimally in polynomial time, it implies that the minimization problem can be solved in polynomial time. Furthermore, \cite{ChauLWZZ19} studied the batch calibrations variant. Calibrations must occur at the same moment on different machines. Additionally, the cost of a batch of calibrations is defined by a non-decreasing function that depends on the number of calibrations occurring at this batch. They showed that this problem could be solved in polynomial time, and gave several faster approximation algorithms for specific cost functions.

A notable statement from \cite{bender2013efficient} attracted our attention:
``{\em As a next step, we hope to generalize our model to capture more
aspects of the actual ISE problem. For example, machines may not
be identical, and calibrations may require machine time. Moreover,
some jobs may not have unit size}".

\paragraph{Our contributions.} In this paper, we investigate the single-machine case without resource augmentation, and we study the complexity status of different variants of the calibration cost minimization problem. In Section~\ref{sec:warm}, we study the problem when the jobs may have arbitrary processing times, and the preemption of the jobs is allowed: the processing of any job may be interrupted and resumed later.  Clearly, by using the optimal algorithm of \cite{bender2013efficient} for unit-time jobs, we can directly obtain a pseudopolynomial-time algorithm by splitting jobs into unit-processing time jobs and replacing every job by a set of unit-time jobs with cardinality equal to the processing time of the job. We propose a polynomial-time algorithm for this variant of the problem.  Then, in Section~\ref{sec:arbitrary_pj}, we study the case of scheduling a set of jobs when $K$ different types of calibrations are available. Each calibration of type $k$ is associated with a length of $T_k$ and a cost $f_k$. The objective is to find a feasible schedule minimizing the total calibration cost. We show that the problem with arbitrary processing times is NP-hard, even when the preemption of the jobs is allowed.

We study the case of unit-time jobs in Section~\ref{sec:act_unit} and propose a polynomial-time algorithm based on dynamic programming. We present an algorithm for a more general setting where calibrations are not instantaneous and require $\lambda$ units of time during which the machine cannot be used. We refer to this period as the \emph{activation length}.


\section{Arbitrary processing times and preemption}\label{sec:warm}

We suppose here that the jobs have arbitrary processing times and that the preemption of the jobs is allowed. An obvious approach to obtain an optimal preemptive schedule is to divide each job $j$ into $p_j$ unit-time jobs with the same release time and deadline as job $j$ and then apply the Lazy Binning (LB) algorithm of~\cite{bender2013efficient} that optimally solves the problem for instances with unit-time jobs. However,  this approach leads to a pseudopolynomial-time algorithm. Here, we propose a more efficient way for the problem. Our method is based on the idea of LB. In the sequel, we suppose without loss of generality that jobs are sorted in non-decreasing order of their deadline, $d_1\leq d_2 \leq \ldots \leq d_n$. Before introducing our algorithm, we briefly recall LB:
at each iteration, a time $t$ (initially 0) is fixed and the (remaining) jobs are scheduled, starting at time $t+1$ using the Earliest Deadline First {\sc (edf)} policy \footnote{In the {\sc edf} policy, at any time $t$, the available jobs are scheduled in order of non-decreasing deadlines.}. 
If a feasible schedule exists (for the remaining jobs), $t$ is updated to $t+1$; otherwise, the next calibration is set to start at time $t$, which is called the current {\em latest-starting-time} of the calibration. Then, the jobs that are scheduled during this calibration interval are removed, and this process is iterated after updating $t$ to $t+T$, where $T$ is the calibration length. The polynomiality of the algorithm for unit-time jobs comes from the observation that the starting time of any calibration is at a distance of no more than  $n$ time-units before any deadline.
In our case, however, i.e., when the jobs have arbitrary processing times, a calibration may start at a distance of at most  $P=\sum_{j=1}^n p_j$ time-units before any deadline.

\begin{definition}\label{def:psi}
Let $\Psi := \bigcup_i \{ d_i-P,d_i-P+1,\ldots, d_i \}$ where $P=\sum_{j=1}^n p_j$.
\end{definition}

\begin{prop}\label{prop:position_calibration}
There exists an optimal solution in which each calibration starts at a time in $\Psi$.
\end{prop}

\begin{proof}
Let $\sigma$ be an optimal solution in which there is at least one calibration that does not start at a time in $\Psi$.
We show how to transform the schedule $\sigma$ into another optimal
solution that satisfies the statement of the proposition without increasing the cost.
See Fig.~\ref{fig:push_calibrations} for an illustration of the transformation.
Let $c_{i'}$ be the first calibration of $\sigma$  that starts at
time $t' \notin \Psi$.
Let $c_{i'},\ldots, c_i$ be the maximum set of consecutive calibrations, i.e.
when a calibration finishes, another starts immediately.
We denote by $c_{i+1}$ the next calibration that is not adjacent to calibration $c_i$.
\begin{figure*}[ht]
\centering
\includegraphics[scale=1.1]{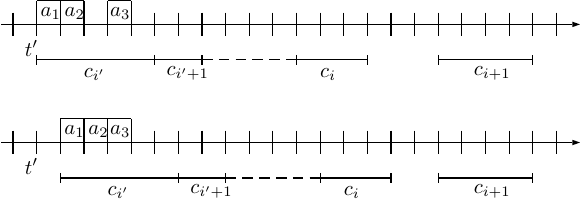}
\caption{Illustration of Proposition~\ref{prop:position_calibration}. The first schedule is an optimal schedule. The second one is obtained after delaying the continuous 
block of calibrations $c_{i'},\ldots, c_i$.}
\label{fig:push_calibrations}
\end{figure*}
We can delay the set of calibrations $c_{i'},\ldots, c_i$ until:
\begin{itemize}
\item either we reach the next calibration $c_{i+1}$,
\item or $c_{i'}$ starts at a time in $\Psi$.
\end{itemize}
Note that this procedure is always possible.
Indeed, since $c_{i'}$ starts at a time that is in a distance more than $P$ from a deadline, it is always possible to delay the scheduled jobs while keeping the feasibility of the schedule. 
In particular, if there are no jobs scheduled when calibration $c_{i'}$ starts, then the execution of jobs can remain unchanged.
Otherwise, there is at least one job scheduled when 
calibration $c_{i'}$ starts. Let  $a_1,\ldots,a_e$ be the continuous block of jobs.
Since the starting time of job $a_1$ is at a distance  (to the left-hand side) more than 
$P$ from a deadline, then all  jobs
can be delayed by one time-unit since no job of this block finishes at its deadline.
Note that after this modification, jobs can be assigned to another calibration.
We repeat the above transformation until we get a schedule satisfying the
statement of the proposition.

\qed\end{proof}

We propose the following algorithm whose idea is based on the LB algorithm:  we first compute the current latest-starting-time
of the calibration
such that no job misses its deadline (this avoids to consider every value in $\Psi$).
The starting time of the calibration depends on some deadline $d_k$. At each iteration, among the remaining jobs, we compute for every deadline the sum of the processing times of all these jobs (or of their remaining parts) having a smaller than or equal deadline, and we subtract it from the current deadline. The current latest-starting-time of the calibration is obtained by choosing the smallest computed value.
Once the starting time of the calibration is
set, we schedule the remaining jobs in the {\sc edf} order until reaching
$d_k$ and we continue to schedule the available jobs until the
calibration interval finishes.
In the next step, we update the processing time of the jobs that have been
processed. We repeat this procedure until all jobs are fully processed. A formal description of the algorithm that we call the Preemptive Lazy Binning (PLB) algorithm is given below (Algorithm~\ref{algo:calibration}).

\begin{algorithm}[thbp]
\begin{algorithmic}[1] 
\STATE Jobs in $\mathcal{J}$ are sorted in non-decreasing order of deadline
\WHILE{$\mathcal{J}\neq \emptyset$}
	\STATE $t \gets \max_{i\in \mathcal{J}}d_i$, $k\gets 0$ // the current latest-starting-time of the calibration
	\FOR{$i\in \mathcal{J}$}
		\IF{$t>d_i-\sum_{j\leq i,j\in \mathcal{J}}p_j$}
			\STATE $t\gets d_i-\sum_{j\leq i,j\in \mathcal{J}}p_j$
			\STATE $k \gets i$
		\ENDIF
	\ENDFOR
	\STATE $u\gets t+ \left\lceil\frac{d_k-t}{T}\right\rceil \times T$
	\STATE Perform calibrations at time $t, t+T, t+2T, \ldots, u-T$.
	\STATE Schedule jobs $\{ j\leq k ~|~j\in \mathcal{J}  \}$ from $t$ to $d_k$ by applying the {\sc edf} policy and remove them from $\mathcal{J}$.
	\STATE Schedule jobs $k+1,\ldots, n$ in $[d_k,u)$ in {\sc edf} order.
	\STATE Let $q_j$ for $j=k+1,\ldots, n$ be the processed quantity in $[d_k,u)$.
	\STATE //Update processing time of jobs
	\FOR{$i=k+1,\ldots, n$}
		\STATE $p_i \gets p_i - q_i$
		\IF{$p_i=0$}
			\STATE $\mathcal{J}\gets \mathcal{J}\setminus i$
		\ENDIF
	\ENDFOR
\ENDWHILE
\end{algorithmic}
\caption{Preemptive Lazy Binning (PLB)}
\label{algo:calibration}
\end{algorithm}

We can prove the optimality of this algorithm using a similar analysis as for the Lazy Binning algorithm in~\cite{bender2013efficient}.

\begin{prop}\label{prop:algo_latest_starting_time}
The schedule returned by Algorithm~PLB is a feasible schedule in
which the starting time of each calibration is the latest possible.
\end{prop}

\begin{proof}
We show that the condition at line 5 in Algorithm~PLB ensures 
that we always obtain a feasible schedule.

We compute the latest-starting-time at each step, and this time is exactly the latest time of the first calibration. By setting a deadline $d_i$, we know that jobs that have a deadline earlier than $d_i$ have to be scheduled before $d_i$, while the other jobs should be scheduled after $d_i$. When we update $t$ for every deadline $d_i$ in the algorithm, we assume that there is no idle time between $d_i-\sum_{j\leq i,j\in \mathcal{J}}p_j$ (i.e. the latest starting time) and $d_i$. Note that if $d_i-\sum_{j\leq i,j\in \mathcal{J}}p_j<0$, then the schedule is not feasible. For the sake of contradiction, suppose that a feasible schedule exists in which some calibration is not started at a time computed by the algorithm. We will show that the starting time of this calibration is not the latest one. Denote this time by $t'$.
Since there is no $i$ for which the starting time of the calibration is  $d_i -\sum_{j\leq i,j\in \mathcal{J}}p_j $, then there is at least one unit of idle time between the starting time of the calibration and some deadline $d_i$. Hence, it is possible to delay all calibrations starting at $t'$ or later, as well as the jobs that were scheduled in $[t',d_i)$ while keeping the {\sc edf} order. This can be done similarly as in the proof of Proposition~\ref{prop:position_calibration}.
\qed\end{proof}

\begin{prop}
Algorithm~$PLB$ is optimal.
\end{prop}

\begin{proof}
It is sufficient to prove that Algorithm~$PLB$ returns the same schedule as $LB$ after splitting all jobs to unit-time jobs. We denote respectively $PLB$ and $LB$ the schedules returned by these algorithms.

Let $t'$ be the first time at which the two schedules differ. The jobs executed before $t'$ are the same in both schedules. Hence, the remaining jobs are the same after $t'$. Two cases may occur:
\begin{itemize}
\item a job is scheduled in $[t',t'+1)$ in $PLB$ but not in $LB$.
This means that the machine is not calibrated at this time slot in the schedule produced by $LB$.
Since the calibrations are the same before $t'$ in both schedules, then a calibration starting at $t'$ is necessary for $PLB$. According to Proposition~\ref{prop:algo_latest_starting_time}, we have a contradiction to the fact that we were looking for the latest-starting-time of the calibration.
\item a job is scheduled in $[t',t'+1)$ in $LB$ but not in $PLB$. This means that there does not exist a feasible schedule starting at $t'+1$ with the remaining jobs. Hence, $PLB$ is not feasible. This case is not possible from Proposition~\ref{prop:algo_latest_starting_time}.
\end{itemize}
\qed\end{proof}

\begin{prop}
Algorithm~$PLB$ has a running time of $O(n^2)$.
\end{prop}

\begin{proof}
We first sort jobs in the non-decreasing order of their deadlines in $O(n\log n)$ time. At each step, we compute the first latest-starting-time of the calibration in $O(n)$ time. Then the scheduling of jobs in the {\sc edf} order takes $O(n)$ time. We also need to update the processing times of the jobs whose execution has been started. This can be done in $O(n)$ time. At each step, we schedule at least one job. Hence, there are at most $n$ steps.
\qed\end{proof}

\section{Arbitrary processing times, preemption and many calibration types}
\label{sec:arbitrary_pj}
In this section, we consider a generalization of the model of~\cite{bender2013efficient} in which there is more than one type of calibration. Every calibration type is associated with a length of $T_i$ and a cost $f_i$. We are also given a set of jobs, each one characterized by its processing time $p_j$, its release time $r_j$, and its deadline $d_j$. Each job can only be scheduled when the machine is calibrated regardless of the calibration type. Our objective is to find a feasible preemptive schedule minimizing the total calibration cost. We prove that the problem is NP-hard.

\begin{prop}    \label{prop_np_hard}
The problem of minimizing the calibration cost is NP-hard for jobs with arbitrary processing times and many types of calibration,
even when the preemption is allowed.
\end{prop}

In order to prove the NP-hardness, we use a reduction from the {\sc Subset Sum} problem which is NP-hard (\cite{johnson1979computers}). In an instance of the {\sc Subset Sum} problem, we are given a set of $n$ items where
each item $j$ is associated with a value $\kappa_j$. We are also given a value $V$. We aim to find a subset of the items that can be summed to $V$ under the assumption that each item may be used once. However, in our proof, we suppose that each item can be used several times, but at most $V/\kappa_j$ times.

\begin{proof}
Let $\Pi$ be the preemptive scheduling problem of minimizing the total calibration cost for a set of $n$ jobs that have arbitrary processing times in the presence of a set of $K$ calibration types.

Given an instance of the {\sc Subset Sum} problem, we construct an instance of the problem $\Pi$ as follows. For each item $j$, we create a calibration length $T_j=\kappa_j$ and of cost $f_j=\kappa_j$. Moreover, we create one job of processing time $V$ that is released at time 0, and its deadline is $V$. We assume that each calibration can be used several times, i.e., each calibration type is duplicated as many as needed.

We claim that the instance of the {\sc Subset Sum} problem is feasible if and only if there is a feasible schedule for problem $\Pi$ of cost $V$.

Assume that the instance of the {\sc Subset Sum} problem is feasible. Therefore, there exists a subset of items $C'$ such that $\sum_{j\in C'}\kappa_j=V$. As mentioned previously, the same item may appear several times. Then we can schedule the unique job, and calibrate the machine according to the items in $C'$ in any order. Since the calibrations allow the job to be scheduled in $[0, V)$, then we get a feasible schedule of cost $V$ for $\Pi$.

In the opposite direction of our claim, assume that there is a feasible schedule for problem $\Pi$ of cost $V$. Let $\mathcal{C}$ be the set of calibrations that have been used in the schedule. Then $\sum_{j\in \mathcal{C}}T_j=V$. Therefore, the items which correspond to the calibrations in $\mathcal{C}$ form a feasible solution for the {\sc Subset Sum} problem. 
\qed\end{proof}

\section{Unit-time jobs, many calibration types and activation length}
\label{sec:act_unit}

We showed previously that the problem is NP-hard when many calibration types are considered even in the case where the calibrations are instantaneous. In this section, we investigate unit-time jobs, and the calibrations are not instantaneous anymore. Every calibration has an activation length $\lambda$ during which the machine cannot process jobs. For feasibility reasons, we allow recalibrating the machine at any point, even when it is already calibrated. However, it is not allowed to calibrate the machine when a job is running. As an example, consider the instance given in Figure~\ref{fig:infeasible}.  The machine has to be calibrated at time 0 and requires $\lambda=3$ units of time for being available for the execution of jobs. At the time $3$, the machine is ready to execute job 1, and it remains calibrated for $T=4$ time units. If we cannot recalibrate an already calibrated machine, then the earliest time at which we can start calibrating the machine is $7$. This would lead to the impossibility of executing job $2$. However, a recalibration at time $4$ would lead to a feasible schedule.
\begin{figure}[ht]
\centering
\includegraphics[scale=1.1]{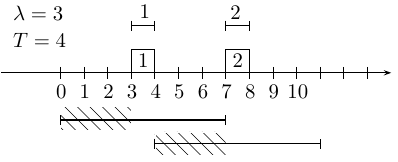}
\caption{An infeasible instance if we cannot recalibrate at any time. We have a single machine, two unit-time jobs, and a unique type of calibration of length $T=4$. The activation length, i.e., the duration that is required for the calibration to be valid, is $\lambda=3$, which is represented by hatched lines in the figure. Job 1 is released at time 3, and its deadline is 4. Job 2 is released at time 7, and its deadline is 8.}
\label{fig:infeasible}
\end{figure}

It is easy to see that the introduction of the activation length into the model makes necessary the extension of the set $\Psi$ of relevant times that we have used in Section~\ref{sec:warm} (Definition~\ref{def:psi}). Indeed, we observe that jobs can be scheduled at a distance larger than $n$ from a release time or a deadline. 
In the worst case, we have to calibrate $n$ times to be able to schedule $n$ jobs. Thus the calibration can start at a time at most $n(\lambda +1)$ time units before a deadline. However, we show in the following that it is not necessary to consider every time in $[d_i-n(\lambda +1),d_i]$ for a fixed $i$ and we define a time set of size polynomial in $n$.


\begin{definition}
Let $\Theta := \bigcup_i \{ d_i-j\lambda-h,~j=0,\ldots, n, ~h=0,\ldots,n \}$.
\end{definition}

Note that this set of relevant times does not depend on the length of the calibration or the activation duration.

\begin{prop}\label{prop:act_position_calibration}
There exists an optimal solution in which each calibration starts at a time in $\Theta$.
\end{prop}
\begin{proof}
We show that it is possible to transform an optimal schedule into another schedule satisfying the statement of the proposition without increasing its cost (a schedule using the same set of calibrations). Let $c_j$ be the last calibration that does not start at a time in $\Theta$. We can delay this calibration until:
\begin{itemize}
\item one job in calibration $c_j$ finishes at its deadline, and hence, it is not possible to delay this calibration anymore. So there is no idle time between the starting time of the calibration $c_j$ and this deadline. Thus the starting time of $c_j$ is in $\Theta$.
\item the current calibration meets another calibration. In this case, we continue to shift the current calibration to the right while this is possible. An overlap between calibration intervals may occur, but as mentioned before, we allow to recalibrate the machine at any time. If we cannot shift to the right anymore, either a job ends at its deadline (and we are in the first case), or there is no idle time between the current calibration and the next one. Since there is at most $n$ jobs and the next calibration starts at a time $d_i-j\lambda-h$ for some $i,j,h$, then the current calibration starts at a time $d_i-(j+1)\lambda-(h+h')$ where $h'$ is the number of jobs scheduled in the current calibration with $h'  +h\leq n$ and $j\leq n-1$.
\end{itemize}
\qed\end{proof}

Moreover, the set of starting times of jobs has also to be extended by considering the activation length $\lambda$.

\begin{definition}
Let $\Phi := \{ t+\lambda + a~|~t\in \Theta,~a=0,\ldots,n \}
\cup \bigcup_i  \{ r_i, r_i+1, \ldots, r_i+n  \}$.
\end{definition}

As for the starting time of calibrations, the worst case happens when we have to recalibrate after the execution of every job.

\begin{prop}\label{prop:act_position_job}
There exists an optimal solution in which the starting times and completion times of jobs are in $\Phi$.
\end{prop}

\begin{proof}
We show that it is possible to transform the schedule into another one respecting the proposition without increasing the cost. First, we suppose that we have an optimal solution in which calibrations occur at time in $\Theta$ (from Proposition~\ref{prop:act_position_calibration}). Suppose now $i$ is the first job that is not scheduled at a time in $\Phi$ in such a solution. The idea is to schedule such a job earlier. Note that the calibrations are fixed in this proof. Two cases may occur:
\begin{itemize}
\item job $i$ meets another job $i'$ (Fig.~\ref{fig:push_jobs}(a)). In this case, we consider the continuous block of jobs $i'',\ldots, i',i$. We assume that at least one job in this block is scheduled at its release time, and job $i$ is at a distance at most $n$ of this release time (because there are at most $n$ jobs). Otherwise, we can shift this block of jobs to the left by one time-unit (Fig.~\ref{fig:push_jobs}(b)). Indeed, this shifting is possible because no job in $\{i'',\ldots,i'\}$ is executed at a starting time of a calibration (if it is the case, job $i$ is in $\Phi$ by definition). Since job $i'$ was in $\Phi$, by moving this block, job $i$ will be scheduled at a time in $\Phi$.
\item job $i$ meets its release time; thus its starting time is in $\Phi$.
\end{itemize}
\qed\end{proof}

\begin{figure}
\centering
\includegraphics[scale=1]{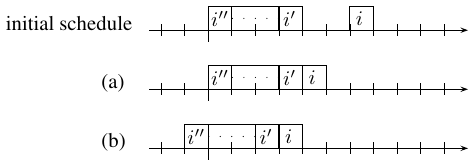}
\caption{Illustration of Proposition~\ref{prop:act_position_job}. The first schedule is the initial schedule in which the job $i$ is not schedule at a time in $\Phi$. The second schedule (a) is the situation when the job $i$ meets another job $i'$ and we are considering the continuous block of jobs $i'',\ldots,i $. The last schedule (b) corresponds to the situation when the block of jobs is scheduled one time-unit earlier.}
\label{fig:push_jobs}
\end{figure}

As mentioned in the introduction, we study the case of scheduling jobs when $K$ different types of calibrations are available. Recall that a calibration of type $k$ is defined by a length $T_k$ and a cost $f_k$. Moreover, calibrations have an activation length of $\lambda$ in which jobs cannot be processed. The cost of a schedule is the sum of the cost of the calibrations that occur (start) within the interval of the schedule.
We are now ready to present the table of our dynamic programming. 

\begin{definition} Let $S(j,u,v)=\{ i~|~ i\leq j \mbox{ and } u \leq r_i <v \}$ be the set of the $j$ first jobs that are released in $[u,v)$.
We define $F(j,u,v,t,k)$ as the minimum cost a schedule such that:
\begin{itemize}
\item the jobs in $S(j,u,v)$ are scheduled during the time-interval $[u,v)$,
\item the first calibration of such a schedule occurs no earlier than $u$,
\item the last calibration is type $k$ and starts at time $t$ for a length of $\lambda+T_k$. This calibration occurs not later than $v$ (and not earlier than $u$).
\end{itemize}
\end{definition}

Note that in the above definition, the time-interval $[t,t+\lambda)$ corresponds to the activation length of the last calibration. Since no job can be scheduled in this interval, it is not relevant to consider when $v$ is in $[t,t+\lambda)$. So in the initialization of the dynamic programming, $v$ should be larger than $t+\lambda$.

Moreover, the end of the schedule should not be after the end of the last calibration, i.e. $v\leq t+\lambda+T_k$. Indeed, let $v'$ be the ending time of the last calibration of the schedule, and suppose we have $v'\leq v$, the interval $[v',v)$ cannot contain any scheduled jobs since it is after the last calibration of the schedule. If a job is released in this interval, it cannot be scheduled according to our definition, therefore the cost of such a schedule is infinite. In the complementary case, if no jobs are released in the interval $[v',v)$, then the schedules ending at $v$ or $v'$ are the same. Thus, in all cases, we have $F(j,u,v',t,k)\leq F(j,u,v,t,k)$.
In the sequel, we consider that $F(j,u,v,t,k)=+\infty$ if $v>t+\lambda + T_k$.

The initialization is as follows:\\
$F(0,u,v,t,k):=f_k,~$ when $u\leq t < t + \lambda\leq v$ and $t+\lambda \leq v \leq t+\lambda + T_k$ for $u,v \in \Phi$, $t\in \Theta$ and $k=1,\ldots,K$.\\
$F(0,u,v,t,k):=+\infty$~otherwise.

We examine the cases depending on whether $r_j$ is in the interval $[u, v)$.
When $r_j \notin [u,v)$ (case 1), then the job $j$ is not scheduled in the schedule associated to $F(j,u,v,t,k)$, so $F(j,u,v,t,k)=F(j-1,u,v,t,k)$.
On the other hand, when $r_j\in [u,v)$ there are two cases: 
\begin{itemize}
\item case 2: job $j$ is scheduled in the last calibration,
\item case 3: job $j$ is not scheduled in the last calibration.
\end{itemize}
In both cases, we need to consider the jobs that are scheduled after the job $j$ but in the same calibration as the job $j$. If we consider that the job $j$ is scheduled at time $u'$, we use $G(j,u',v')$, where $v'$ is not after the end of the calibration, to represent the jobs scheduled after the job $j$.
%
%

\begin{definition}
We define $G(j,u,v)$ to be the schedule of the jobs in $S(j,u,v)$ such that they are scheduled in the interval $[u,v)$. 
\end{definition}
Here, we aim to schedule all jobs in $S(j,u,v)$ within the interval $[u,v)$. If the schedule is feasible, then the cost is 0, while if there is no feasible schedule, then the cost is infinite. We only need to schedule the jobs in {\sc edf} order to see whether the schedule is feasible.

\begin{prop}\label{prop:dp_act_calibration}
One has $F(j,u,v,t,k)=F'$ where
\begin{align*}
&\mbox{Case 1: } r_j\notin [u,v)\\
&F'=F(j-1,u,v,t,k)\\
&\mbox{Case 2: job $j$ is scheduled in the last calibration}\\
&F'=\min_{\substack{
	u'\in \Phi,~r_j\leq u'<t+\lambda+T_{k}\\
	t+\lambda \leq u' < v\\
	u'<d_j}}
	\left\{
	\begin{array}{r}
	F(j-1,u,u',t,k) \\ + G(j-1,u'+1,v)
	\end{array}
	\right\}
	\\
&\mbox{Case 3: job $j$ is not scheduled in the last calibration: }\\
&F'=\min_{\substack{
	u' , v' \in \Phi,~t'\in \Theta, 1\leq k'\leq K\\
	v'\leq t\\
	~r_j\leq u'<v' \leq t'+T_{k'}+\lambda\\
	t'+\lambda \leq u' < d_j
	}}
	\left\{
\begin{array}{c}
F(j-1,u,u',t',k')\\
+G(j-1,u'+1,v')\\
+F(j-1,v',v,t,k)
\end{array}
\right\}
\end{align*}
%
\end{prop}

The objective function for our problem is
$\min_{t\in \Theta, 1\leq k \leq K,v\in \Phi, v \geq \max_i r_i } $ $F(n,\min_i r_i,v,t,k)$.

\begin{figure*}[h!]
\centering
\includegraphics[scale=1.1]{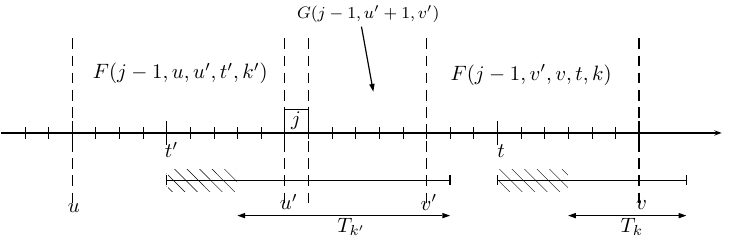}
\caption{Illustration of Proposition~\ref{prop:dp_act_calibration} where job $j$ is scheduled at time $u'$. We can divide the associated schedule in the interval $[u,v)$ into three sub schedules with the respective interval $[u,u')$, $[u'+1,v')$ and $[v',v)$.}
\label{fig:dp_calibration}
\end{figure*}

\begin{proof}
When $r_j\notin [u,v)$, we have necessarily $F(j,u,v,t,k)=F(j-1,u,v,t,k)$. In the following, we suppose that $r_j\in [u,v)$ which falls into the last two cases. 

\noindent
\textbf{We first prove that $F(j,u,v,t,k) \leq F'$ (feasibility).}

Case 2: We consider a schedule $S_1$ that realizes $F(j-1,u,u',t,k)$, a schedule $S_2$ that realizes  $G(j-1,u'+1,v)$. We build a schedule as follows: from time $u$ to time $u'$ use $S_1$, then execute job $j$ in $[u',u'+1)$, then from  $u'+1$ to $v$ we use $S_2$. Moreover, it contains all jobs in $\{i~|~i \leq j \mbox{ and } u \leq r_i < v\}$. 

Case 3: We consider a schedule $S_1$ that realizes $F(j-1,u,u',t',k')$, a schedule $S_2$ that realizes  $G(j-1,u'+1,v')$ and a schedule $S_3$ that realizes $F(j-1,v',v,t,k)$. We build a schedule as follows: from time $u$ to time $u'$ use $S_1$, then execute job $j$ in $[u',u'+1)$, then from  $u'+1$ to $v'$ we use $S_2$ and finally from  $v'$ to time $v$ we use $S_3$. Moreover, it contains all jobs in $\{i~|~i \leq j \mbox{ and } u \leq r_i < v\}$. 

Note that in both case, the interval $[u',v')$ is covered by the last calibration of $S_1$ and since the first calibration in $S_3$ does not begin before $v'$, then we have a feasible schedule.

So $F(j,u,v,t,k) \leq F'$.

\noindent
\textbf{We now prove that $F(j,u,v,t,k)\geq F'$ (optimality).}

Since $j \in \{i~|~i \leq j \mbox{ and } u \leq r_i < v\}$, job $j$ is scheduled in all schedules that realize $F(j,u,v,t,k)$.

Among such schedules, let $\mathcal{X}$ denote the schedule of $F(j,u,v,t,k)$ in which the starting time of job $j$ is maximal ($u'$ is maximal), and then $v'$ is maximal. We claim that all jobs in $\{i \leq j, u \leq r_i < v\}$ that are released before $u'$ are completed at $u'$. If it is not the case, we could swap the execution of such a job with the job $j$, getting in this way a feasible schedule with the same cost as before. More formally, let $i$ be a job with $\{i~|~i \leq j, u \leq r_i < u'\}$ that is scheduled after $u'+1$. We can swap the execution of job $i$ with job $j$, which results in a feasible schedule since job $j$ has larger deadline than the job $i$, and the job $i$ is released before $u'$. This will contradict the fact that the starting time of job $j$ is maximal.
Similarly, it can be noticed that no job in $S(j-1,u,v)$ is released at time $u'$. Otherwise, this job would be scheduled at time $u'$ since its deadline is smaller than the deadline of the job $j$, which also contradicts the fact that $u'$ is maximal.

Moreover, we claim that all jobs released in $[u'+1,v')$ are completed before $v'$. Suppose on the contrary that there is a job $i\in S(j-1,u'+1,v')$ that is scheduled after $v'$. It means that $v'$ is smaller since this job does belong to the schedule after $v'$. It contradicts the fact that $v'$ was maximal.

In the case 2, we consider a schedule $S_1$ that realizes $F(j-1,u,u',t,k)$, a schedule $S_2$ that realizes $G(j-1,u'+1,v)$. Then, the restriction of $S_1$ in the schedule $\mathcal{X}$ to $[u, u')$ will be a schedule that meets all constraints related to $F(j-1,u,u',t,k)$. Hence its cost is greater than $F(j-1,u,u',t,k)$. Similarly, the restriction of $S_2$ in the schedule $\mathcal{X}$ to $[u'+1, v)$ is a schedule that meets all constraints related to $G(j-1,u'+1,v')$.

Similarly in the case 3, we consider a schedule $S_1$ that realizes $F(j-1,u,u',t',k')$, a schedule $S_2$ that realizes $G(j-1,u'+1,v')$ and a schedule $S_3$ that realizes $F(j-1,v',v,t,k)$. Then, the restriction of $S_1$ in the schedule $\mathcal{X}$ to $[u, u')$ will be a schedule that meets all constraints related to $F(j-1,u,u',t',k')$. Hence its cost is greater than $F(j-1,u,u',t',k')$. Similarly, the restriction of $S_2$ in the schedule $\mathcal{X}$ to $[u'+1, v')$ is a schedule that meets all constraints related to $G(j-1,u'+1,v')$ and the restriction of $S_3$ in the schedule $\mathcal{X}$ to $[v', v)$ is a schedule that meets all constraints related to $F(j-1,v',v,t,k)$.

Finally, $F(j,u,v,t,k)\geq F'$.
\qed\end{proof}

\begin{prop}
The problem of minimizing the total calibration cost with arbitrary calibration
lengths, activation length and unit-time jobs can be solved in time
\nocolor{$O(n^{19}K^2)$. }
\end{prop}

\begin{proof}
This problem can be solved with the dynamic program in Proposition~\ref{prop:dp_act_calibration}. Recall that the table is $F(j,u,v,t,k)$ where $j\in \{0,\ldots,n\}$, $u,v\in \Phi$, $t\in \Theta$ and $k\in \{1,\ldots,K\}$.
The size of both sets $\Theta$ and $\Phi$ is $O(n^3)$. Indeed, by rewriting the set $\Phi$, we have
\begin{align*}
\Phi =  &\{ r_i, r_i+1\ldots, r_i+n ~|~ 1\leq i\leq n \} \\
&\qquad \cup \{ t+a~|~t\in \Theta,~a=0,\ldots,n \}\\
= &\bigcup_i  \{ r_i, r_i+1\ldots, r_i+n  \} \\
&\bigcup_i \left\{ \begin{array}{c}
d_i-j\lambda-k+a,~j=0,\ldots, n\\
k=0,\ldots,n,~ a=0,\ldots,n
\end{array} \right\}\\
=&\bigcup_i  \{ r_i, r_i+1\ldots, r_i+n  \} \\
&\bigcup_i \{ d_i-j\lambda+k,~j=0,\ldots, n, ~k=-n,\ldots,n\}
\end{align*}
So, the size of the table is $O(n^{10}K)$. When each value of the table is fixed, the minimization is over the values $u', v'\in \Theta$, $t'\in \Phi$ and $k'\in \{1,\ldots,K  \}$, so the running time is $O(n^9K)$. 
Recall that the objective function is 
$\min_{t\in \Theta, 1\leq k \leq K,v\in \Phi, v \geq \max_i r_i } $ $F(n,\min_i r_i,v,t,k)$.
Thus, the overall time complexity is $O(n^{19}K^2)$.
\qed\end{proof}

Note that when there is no feasible schedule, the dynamic programming will return $+\infty$.

\section{Conclusion}
We considered different extensions of the model introduced by 
\cite{bender2013efficient}. We proved that the problem of minimizing the total calibration-cost on a single machine could be solved in polynomial time for the case of jobs with arbitrary processing times when the preemption is allowed. Then we proved that the problem becomes NP-hard for arbitrary processing times and many calibration types, even if the preemption of jobs is authorized. Finally, we considered the case with many calibration types, not instantaneous calibrations and unit-time jobs, proving that the problem can be solved in polynomial time by using dynamic programming techniques. An interesting question is whether it is possible to find a lower time-complexity algorithm for solving this version of the problem, either optimally or an approximation. Of course, it would be of great interest to study the case where more than one machines are available. Recall that the complexity of the simple variant studied by \cite{bender2013efficient} remains unknown for the multiple machines problem.

\bibliographystyle{spbasic}      
\bibliography{biblio}

\end{document}